\documentclass[12pt,twoside,reqno]{amsart}
\usepackage{pifont}
\usepackage{amsmath}
\usepackage{amsthm}
\usepackage{amsfonts}
\usepackage{amssymb}
\usepackage{latexsym}
\usepackage{amstext}
\usepackage{array}
\usepackage{graphicx}

\date{}
\allowdisplaybreaks[4] \footskip=15pt
\renewcommand{\uppercasenonmath}[1]{}

\newtheorem{Theorem}{Theorem}

\newtheorem{Lemma}{Lemma}
\newtheorem{Proposition}{Proposition}
\newtheorem{Conjecture}{Conjecture}

\theoremstyle{definition}
\newtheorem{Definition}{Definition}

\usepackage{amscd,amsthm,amsfonts,amssymb,amsmath}
\usepackage{fullpage}
\usepackage[all]{xy}

    \numberwithin{equation}{section}

\input txdtools
\begin{document}
\begin{center}

{\large  \bf A Proof of a Conjecture About a Class of  Near Maximum Distance Separable Codes}

\vskip 0.8cm
{\small Wei Lu$^1$ $\cdot$ Xia Wu$^1$ \footnote{Supported by NSFC (Nos. 11971102, 11801070).

  }}\\

{\small $^1$School of Mathematics, Southeast University, Nanjing
210096, China}\\
{\small e-mail:
 luwei1010@139.com, wuxiadd1980@163.com}\\
{\small $^*$Corresponding author. (e-mail: wuxiadd1980@163.com)}
\vskip 0.8cm
\end{center}

{\bf Abstract:} In this paper, we  completely determine the number   of solutions to $  \operatorname{Tr}^{q^2}_q(bx+b)+c=0,  x\in \mu_{q+1}\backslash \{-1\}$   for all $b\in \mathbb{F}_{q^2}, c\in\mathbb{F}_{q}$. As an application, we can give the    weight distributions of a class of  linear codes, and    give  a  completely  answer to   a recent conjecture about a class of NMDS codes proposed by Heng.


\section{Introduction}
Let $p$ be a prime, $q$ a power of $p$, and $\mathbb{F}_q$ the finite field with $q$ elements. Let $\mathbb{F}_q^*$ be  the multiplicative cyclic group of non-zero elements of  $\mathbb{F}_q$.
  The Singleton bound of an $[n,k,d]_q$ linear code is given by $d\leq n-k+1.$ A linear code with parameters $[n,k,n-k+1]_q$ is called an \emph{maximum distance separable} (MDS for short) code. A linear code with parameters $[n,k,n-k]_q$ is said to be \emph{almost maximum distance separable} (AMDS for short). A code is called \emph{near maximum distance separable} (NMDS for short) if both the code and its dual are almost maximum distance separable. The readers are referred to \cite{DT2020, DL1995, DDK1997, HP2003, KL1984, LH1997, WH2021, W2021, WHL2021} for researches on MDS, AMDS and NMDS codes.

NMDS codes have nice properties and many applications \cite{DL1995,DL2000,FW1997,TD2013}. For example, NMDS codes can be used to construct $t$-designs. The first NMDS code was the $[11,6,5]_3$ ternary Golay code discovered in 1949 by Golay \cite{G1949}. This ternary code holds $4$-designs, and its extended code holds a Steiner system $S(5, 6, 12)$ with the largest strength known. In \cite{DT2020}, Ding and Tang   presented an infinite family of NMDS codes over $\mathbb{F}_{3^m}$ holding an infinite family of $3$-designs and an infinite family of NMDS codes over $\mathbb{F}_{2^{2m}}$ holding an infinite family of $2$-designs  . In \cite{TD2021}, Tang and Ding   presented a family of NMDS codes over $\mathbb{F}_{2^{2m+1}}$ holding an infinite family of $4$-designs, and a family of NMDS codes over $\mathbb{F}_{2^{2m}}$ holding an infinite family of $3$-designs.

In \cite{H2023}, Heng   constructed several classes of linear codes with five families of almost difference sets, and got two families of NMDS codes: one is a family of $[q+1,3,q-2]_q$ NMDS codes in \cite[Theorem 7.5]{H2023},  where $q=2^e$ and $e$ is odd;  the other is a family of $[q+2,3,q-1]_q$ NMDS codes  for odd $q$ in \cite[Theorem 7.7]{H2023}, where $q$ is an odd prime power. Besides these two families of NMDS codes, Heng left another family of NMDS codes in a conjecture \cite[Conjecture 1]{H2023}. One of the objectives of this paper is  to prove this conjecture.

Let $\mu_{q+1}=\{x\in \mathbb{F}_{q^2}: x^{q+1}=1\}$ and $D=\mu_{q+1}\backslash \{-1\}.$ Let $\operatorname{Tr}^{q^2}_q$  be the trace function from $\mathbb{F}_{q^2}$ to $\mathbb{F}_{q}$ defined by $\operatorname{Tr}^{q^2}_q(x)=x+x^q$.  For $b\in \mathbb{F}_{q^2}$ and $c\in\mathbb{F}_{q},$ define the codeword
\begin{equation*}
  c(b,c):=((\operatorname{Tr}^{q^2}_q(bx+b)+c)_{x\in D},-\operatorname{Tr}^{q^2}_q(b)),
\end{equation*}
and the linear code
\begin{equation}\label{code}
  \widetilde{\overline{C_D}}=\{c(b,c):b\in \mathbb{F}_{q^2},c\in\mathbb{F}_{q}\}.
\end{equation}
In \cite{H2023}, Heng conjectured that
\begin{Conjecture} \cite[Conjecture 1]{H2023}\label{conjecture}
Let $D$ and $\widetilde{\overline{C_D}}$ be defined as above. If $q>2$, then $\widetilde{\overline{C_D}}$ is a $[q+1,3,q-2]_q$ NMDS code.
\end{Conjecture}
In this paper, we will give the weight distribution of $\widetilde{\overline{C_D}}.$  Then we get the following   answer to Conjecture \ref{conjecture}:
\begin{enumerate}
  \item If $q=3,5$, then $\widetilde{\overline{C_D}}$ is a $[q+1,3,q-1]_q$ MDS code.
  \item If $q\ne3,5$, then $\widetilde{\overline{C_D}}$ is a $[q+1,3,q-2]_q$ NMDS code.
\end{enumerate}

In order to give the weight distribution of $\widetilde{\overline{C_D}},$ we need   to solve the following equations: let $b\in \mathbb{F}_{q^2}$ and $c\in\mathbb{F}_{q}$, the equation $E(b,c)$ about $x$ is
\begin{equation}\label{equation}
  \operatorname{Tr}^{q^2}_q(bx+b)+c=0,  x\in \mu_{q+1}\backslash \{-1\}.
\end{equation}
Let $N(b,c)$ be the number  of solutions to the equation $E(b,c)$ in $\mu_{q+1}\backslash \{-1\}.$ The main content of this paper is to give the  explicit formula of $N(b,c)$. Moverover, we can give an approach to find  these
solutions.

The rest of this paper is organized as follows. In Section \ref{Preliminaries}, we introduce some basic results about  the trace functions, the norm functions over finite fields, and quadratic forms over odd characteristic finite fields. In Section \ref{even characteristic cases}, we will solve the equations $E(b,c)$ in the  even characteristic cases.  In Section \ref{odd characteristic cases}, we will solve the equations $E(b,c)$ in the  odd characteristic cases. In Section \ref{the weight distribution}, we give  the weight distributions and then give an answer to Conjecture \ref{conjecture}.     In Section \ref{concluding remarks}, we conclude this paper.

\section{Preliminaries}\label{Preliminaries}
\subsection{Trace Functions and   Norm Functions over Finite Fields}

Let $r$ be a prime power and $m$ a positive integer. The trace function $\operatorname{Tr}^{r^m}_r$  from $\mathbb{F}_{r^m}$ to $\mathbb{F}_{r}$ is defined by
\begin{equation*}
\operatorname{Tr}^{r^m}_r(x)=x+x^r+\cdots+x^{r^{m-1}},
\end{equation*}
and the  norm function  $\operatorname{N}^{r^m}_r$  from $\mathbb{F}_{r^m}$ to $\mathbb{F}_{r}$ is defined by
\begin{equation*}
\operatorname{N}^{r^m}_r(x)=x\cdot x^r\cdot\cdots\cdot x^{r^{m-1}}=x^{\frac{r^m-1}{r-1}}.
\end{equation*}
\begin{Lemma}\cite[2.23. Theorem]{LH1997}\label{trace}
The trace function $\operatorname{Tr}^{r^m}_r$  satisfies the following properties:
\begin{description}
  \item[(i)] $\operatorname{Tr}^{r^m}_r(\alpha+\beta)=\operatorname{Tr}^{r^m}_r(\alpha)+\operatorname{Tr}^{r^m}_r(\beta)$ for all $\alpha,\beta\in \mathbb{F}_{r^m};$
  \item[(ii)]  $\operatorname{Tr}^{r^m}_r(k\alpha)=k\operatorname{Tr}^{r^m}_r(\alpha)$ for all $k\in \mathbb{F}_{r}, \alpha\in \mathbb{F}_{r^m};$
  \item[(iii)]    $\operatorname{Tr}^{r^m}_r$ is a linear transformation from $\mathbb{F}_{r^m}$ onto $\mathbb{F}_{r}$, where both $\mathbb{F}_{r^m}$ and $\mathbb{F}_{r}$ are viewed as vector spaces over $\mathbb{F}_{r}$;
  \item[(iv)]  $\operatorname{Tr}^{r^m}_r(k)=mk$ for all $k\in \mathbb{F}_{r};$
  \item[(v)] $\operatorname{Tr}^{r^m}_r(\alpha^r)=\operatorname{Tr}^{r^m}_r(\alpha)$ for all $\alpha\in \mathbb{F}_{r^m}.$
\end{description}
\end{Lemma}

The following lemma  is important  in solving the equations $E(b,c)$ for both even characteristic cases and odd characteristic cases.
\begin{Lemma}\label{Lemma2}
Let $p$ be a prime number and $q$ a power of $p.$ Define $\operatorname{Ker}(\operatorname{Tr}^{q^2}_q):=\{x\in \mathbb{F}_{q^2}: \operatorname{Tr}^{q^2}_q(x)=0\}$.
\begin{description}
  \item[(i)] If $p=2$, then there exists $\alpha\in \mathbb{F}_{q^2}\backslash\mathbb{F}_{q}$ such that $\operatorname{Tr}^{q^2}_q(\alpha)=1$, and $\operatorname{Ker}(\operatorname{Tr}^{q^2}_q)=\mathbb{F}_{q}.$ We also have $\alpha^{q+1}=\operatorname{N}^{q^2}_q(\alpha)\in\mathbb{F}_{q}.$
  \item[(ii)] If $p\ne 2$, then there exists $\alpha\in \mathbb{F}_{q^2}\backslash\mathbb{F}_{q}$ such that $\operatorname{Tr}^{q^2}_q(\alpha)=0$, and $\operatorname{Ker}(\operatorname{Tr}^{q^2}_q)=\mathbb{F}_{q}\alpha.$ We also have $\alpha^{q+1}=\operatorname{N}^{q^2}_q(\alpha)\in\mathbb{F}_{q}.$
\end{description}
\end{Lemma}

The following lemma  will be used in solving the equations $E(b,c)$ for even characteristic cases.
\begin{Lemma}\cite[3.79. Corollary]{LH1997}\label{Lemma3}
Let $p=2$, $q$ a power of $2$ and $\delta\in\mathbb{F}_{q}.$ Then the number  of solutions to the equation $x^2+x+\delta=0$ in $\mathbb{F}_{q}$ is equal to
  \begin{equation*}
 \begin{cases}
0,&{\text{if}}\  \operatorname{Tr}^{q}_2(\delta)=1, \\
{2,}&{\text{if}}\  \operatorname{Tr}^{q}_2(\delta)=0.
\end{cases}
\end{equation*}
\end{Lemma}

\subsection{Quadratic Forms over Odd Characteristic Finite Fields}
Quadratic forms  will be used in solving the equations $E(b,c)$ for odd characteristic cases.  For more information about quadratic forms over odd characteristic finite fields, see \cite[pp. 278--289]{LH1997}.

 In this subsection, let $p$ be an odd prime number and $q$ a power of $p.$ Let  $n$ be a positive integer, and
\begin{equation*}
f(x_1,\dots,x_n)=\sum_{i,j=1}^{n}a_{ij}x_ix_j, {\text{with}}\ a_{ij}=a_{ji}
\end{equation*}
be a quadratic form over $\mathbb{F}_{q}$. The matrix
\begin{equation*}
A_f=(a_{ij})_{n \times n}
\end{equation*}
is called the \emph{coefficient matrix} of $f$. We may define the \emph{determinant} of $f$ is
\begin{equation*}
  \operatorname{det}_f=\operatorname{det}(A_f).
\end{equation*}
We call $f$ is  \emph{nondegenerate }if $\operatorname{det}_f\ne0.$
\begin{Definition}
\begin{enumerate}
  \item \cite[6.22. Definition]{LH1997} The integer-valued function $\nu$  on  $\mathbb{F}_{q}$ is defined by
     \begin{equation*}
   \nu(\Delta)=\begin{cases}
   q-1 &, {\text{if}}\  \Delta=0, \\
  \ -1 &, {\text{if}}\  \Delta\in \mathbb{F}_{q}^*.
   \end{cases}
    \end{equation*}
  \item  The integer-valued function $\mu$  on  $\mathbb{F}_{q}$ is defined by
     \begin{equation*}
   \mu(\Delta)=\begin{cases}
   1 &, {\text{if}}\  \Delta=0, \\
    0 &, {\text{if}}\  \Delta\in \mathbb{F}_{q}^*.
   \end{cases}
    \end{equation*}
  \item The integer-valued function $\eta$  on  $\mathbb{F}_{q}$ is defined by
     \begin{equation*}
   \eta(\Delta)=\begin{cases}
   0 &, {\text{if}}\  \Delta=0, \\
  -1 &, {\text{if}}\  \Delta\  {\text{is   not a square in}} \ \mathbb{F}_q^{*},\\
  1 &, {\text{if}}\  \Delta\  {\text{is  a square in}} \ \mathbb{F}_q^{*}.
   \end{cases}
    \end{equation*}
    In fact, $\eta$ is the   quadratic character of $\mathbb{F}_{q},$ and $\eta(-1)=(-1)^{(q-1)/2}.$
\end{enumerate}
\end{Definition}

From now on, we will consider the number of solutions of some quadratic  equations. Let $\Delta\in \mathbb{F}_{q}$ and $S$ be a subset of $\mathbb{F}_{q}^n.$
We define
\begin{equation*}
  N(f=\Delta; S)
\end{equation*}
to be the number of solutions of the   equation
\begin{equation*}
  f(x_1,\dots,x_n)=\Delta, (x_1,\dots,x_n)\in S.
\end{equation*}
We need the following results. It is convenient to distinguish the cases of even and odd $n$.
\begin{Lemma}\cite[6.26. Theorem]{LH1997}\label{theorem 6.26}
Let $f$ be a nondegenerate quadratic form over $\mathbb{F}_{q},$ $q$ odd, in an even number $n$ of indeterminates. Then for $\Delta\in \mathbb{F}_{q},$ the number  $N(f=\Delta;  \mathbb{F}_{q}^n)$ of
solutions of the equation $f(x_1,\dots,x_n)=\Delta$   in $\mathbb{F}_{q}^n$ is
\begin{equation*}
  q^{n-1}+\nu(\Delta)q^{(n-2)/2}\eta((-1)^{n/2}\operatorname{det}_f).
\end{equation*}
\end{Lemma}
\begin{Lemma}\cite[6.27. Theorem]{LH1997}\label{theorem 6.27}
Let $f$ be a nondegenerate quadratic form over $\mathbb{F}_{q},$ $q$ odd, in an odd number $n$ of indeterminates. Then for $\Delta\in \mathbb{F}_{q},$ the number  $N(f=\Delta;  \mathbb{F}_{q}^n)$ of
solutions of the equation $f(x_1,\dots,x_n)=\Delta$   in $\mathbb{F}_{q}^n$ is
\begin{equation*}
  q^{n-1}+q^{(n-1)/2}\eta((-1)^{(n-1)/2}\Delta\operatorname{det}_f).
\end{equation*}
\end{Lemma}

Now we apply  these two lemmas to the following   quadratic forms. Let
\begin{equation}\label{quadratic form}
g(x_1,x_2,x_3)=x_2^2-x_3^2-4x_1x_3
\end{equation}
be a quadratic form over $\mathbb{F}_{q}.$ The coefficient matrix of $g$ is
\begin{equation}
A_g=
\begin{pmatrix}
0 & 0 & -2 \\
0 & 1 & 0 \\
-2 & 0 & -1
\end{pmatrix},
\end{equation}
and the determinant of $g$ is
\begin{equation*}
  \operatorname{det}_{g}=\operatorname{det}(A_g)=-4.
\end{equation*}

The following lemma  will be used in solving the equations $E(b,c)$ for odd characteristic cases.
\begin{Lemma}\label{Lemma6}
Let $p$ be an odd prime number and $q$ a power of $p.$ Let $\Delta\in \mathbb{F}_{q}$ and $g(x_1,x_2,x_3)=x_2^2-x_3^2-4x_1x_3$. Then
\begin{description}
  \item[(i)]   The number $N(g=\Delta; \mathbb{F}_{q}^*\times\mathbb{F}_{q}\times\mathbb{F}_{q}^*)$ of solutions of the   equation
\begin{equation*}
  g(x_1,x_2,x_3)=\Delta, (x_1,x_2,x_3)\in \mathbb{F}_{q}^*\times\mathbb{F}_{q}\times\mathbb{F}_{q}^*
\end{equation*}
  is

      \begin{equation*}
  \begin{cases}
   q^2-3q+2 &, {\text{if}}\  \Delta=0, \\
   q^2-2q+1 &, {\text{if}}\  \Delta\  {\text{is   not a square in}} \ \mathbb{F}_q^{*},\\
   q^2-2q+3 &, {\text{if}}\  \Delta\  {\text{is  a square in}} \ \mathbb{F}_q^{*}.
   \end{cases}
    \end{equation*}
  \item[(ii)] The number $N(g=\Delta; \{0\}\times\mathbb{F}_{q}^*\times\mathbb{F}_{q}^*)$ of solutions of the   equation
\begin{equation*}
  g(x_1,x_2,x_3)=\Delta, (x_1,x_2,x_3)\in \{0\}\times\mathbb{F}_{q}^*\times\mathbb{F}_{q}^*
\end{equation*}
  is

      \begin{equation*}
  \begin{cases}
   2q-2 &, {\text{if}}\  \Delta=0, \\
   q-2+\eta(-1) &, {\text{if}}\  \Delta\  {\text{is   not a square in}} \ \mathbb{F}_q^{*},\\
   q-4-\eta(-1) &, {\text{if}}\  \Delta\  {\text{is  a square in}} \ \mathbb{F}_q^{*}.
   \end{cases}
    \end{equation*}
\end{description}
\end{Lemma}
\begin{proof}
\begin{description}
  \item[(i)] By Lemma \ref{theorem 6.27}, $N_1=N(g=\Delta; \mathbb{F}_{q}^3)$ is equal to
\begin{equation*}
  q^2+q\eta(\Delta).
\end{equation*}
 By Lemma \ref{theorem 6.26}, $N_2=N(g=\Delta; \{0\}\times\mathbb{F}_{q}\times\mathbb{F}_{q})=N(x_2^2-x_3^2=\Delta; \mathbb{F}_{q}^2)$ is equal to
\begin{equation*}
  q+\nu(\Delta).
\end{equation*}
 By Lemma \ref{theorem 6.27}, $N_3=N(g=\Delta; \mathbb{F}_{q}\times\mathbb{F}_{q}\times\{0\})=N(x_2^2=\Delta; x_2\in \mathbb{F}_{q},  x_3\in \mathbb{F}_{q})=qN(x_2^2=\Delta; x_2\in \mathbb{F}_{q})$ is equal to
\begin{equation*}
  q(1+\eta(\Delta)).
\end{equation*}
 By Lemma \ref{theorem 6.27}, $N_4=N(g=\Delta; \{0\}\times\mathbb{F}_{q}\times\{0\})=N(x_2^2=\Delta; x_2\in \mathbb{F}_{q})$ is equal to
\begin{equation*}
  1+\eta(\Delta).
\end{equation*}

Combining above all, we have  $N(g=\Delta; \mathbb{F}_{q}^*\times\mathbb{F}_{q}\times\mathbb{F}_{q}^*)=N_1-N_2-N_3+N_4$ is equal to
  \begin{equation*}
   q^2-2q+1-\nu(\Delta)+\eta(\Delta).
  \end{equation*}
  \item[(ii)] By Lemma \ref{theorem 6.26}, $N_5=N(g=\Delta; \{0\}\times\mathbb{F}_{q}\times\mathbb{F}_{q})=N(x_2^2-x_3^2=\Delta; \mathbb{F}_{q}^2)$ is equal to
\begin{equation*}
  q+\nu(\Delta).
\end{equation*}
 By Lemma \ref{theorem 6.27}, $N_6=N(g=\Delta; \{0\}\times\{0\}\times\mathbb{F}_{q})=N(-x_3^2=\Delta; \mathbb{F}_{q})$ is equal to
\begin{equation*}
  1+\eta(-\Delta).
\end{equation*}
 By Lemma \ref{theorem 6.27}, $N_7=N(g=\Delta; \{0\}\times\mathbb{F}_{q}\times\{0\})=N(x_2^2=\Delta; x_2\in \mathbb{F}_{q})$ is equal to
\begin{equation*}
  1+\eta(\Delta).
\end{equation*}
 By calculation, $N_8=N(g=\Delta; \{0\}\times\{0\}\times\{0\})$ is equal to
\begin{equation*}
  \mu(\Delta).
\end{equation*}
Combining above all, we have  $N(g=\Delta; \{0\}\times\mathbb{F}_{q}^*\times\mathbb{F}_{q}^*)=N_5-N_6-N_7+N_8$ is equal to
  \begin{equation*}
    q-2+\nu(\Delta)+\mu(\Delta)-\eta(\Delta)(1+\eta(-1)).
  \end{equation*}
\end{description}
\end{proof}

\section{even characteristic cases} \label{even characteristic cases}
In this section, we will solve the equations $E(b,c)$ for all $b\in \mathbb{F}_{q^2}, c\in\mathbb{F}_{q}$ in the  even characteristic cases.
\begin{Theorem}\label{theorem even}
Let $p=2$ and $q$ a power of $p$. Fix $\alpha\in \mathbb{F}_{q^2}\backslash\mathbb{F}_{q}$ such that $\operatorname{Tr}^{q^2}_q(\alpha)=1$. Let $b_1, b_2, c\in\mathbb{F}_{q}$ and $b=b_1+b_2\alpha\in\mathbb{F}_{q^2}$. When $b_2\ne c,$ define
\begin{equation*}
\delta=\frac{c^2\alpha^{q+1}+cb_1}{(b_2+c)^2}.
\end{equation*}
 Then the number $N(b,c)$ of solutions to the equation $E(b,c)$ in $\mu_{q+1}\backslash \{-1\}$  is equal to
 \begin{description}
   \item[(i)] $q$, if $b=0,  c=0;$
   \item[(ii)] 0, if $b=0, c\ne0;$
   \item[(iii)] 0, if $b\ne0, c=0,b_2=0;$
   \item[(iv)] 1, if $b\ne0, c=0,b_2\ne0;$
   \item[(v)] 1, if $b\ne0, c\ne0,b_2=c;$
   \item[(vi)] 0, if $b\ne0, c\ne0,b_2\ne c, \operatorname{Tr}^{q}_2(\delta)=1;$
   \item[(vii)] 2, if $b\ne0, c\ne0,b_2\ne c, \operatorname{Tr}^{q}_2(\delta)=0.$
 \end{description}
\end{Theorem}
\begin{proof}
When $b=0,$ the results are trivial. Next, we let $b\ne0.$  Since $c\in \mathbb{F}_{q}$, by Lemma \ref{trace} (ii), we have $c=\operatorname{Tr}^{q^2}_q(c\alpha).$  We can  rewrite  the equation $E(b,c)$ as
\begin{equation*}
  \operatorname{Tr}^{q^2}_q(bx+b+c\alpha)=0.
\end{equation*}
By Lemma \ref{Lemma2} (i), there exists $y\in\mathbb{F}_{q}$ such that
\begin{equation*}
bx+b+c\alpha=y.
\end{equation*}
Then
\begin{equation}\label{even x}
  x=\frac{y+b+c\alpha}{b},
\end{equation}
and
\begin{equation*}
  x^q=\frac{y+b^q+c\alpha^q}{b^q}.
\end{equation*}
Since $x\in \mu_{q+1}\backslash \{-1\}$, we have
\begin{equation*}
  1=x^{q+1}=x\cdot x^q=\frac{y+b+c\alpha}{b}\cdot \frac{y+b^q+c\alpha^q}{b^q}.
\end{equation*}
we rewrite the above equality as
\begin{equation*}
  y^2+y(c\alpha+c\alpha^q+b+b^q)+c^2\alpha^{q+1}+c(\alpha b^q+\alpha^{q}b)=0.
\end{equation*}
Since
\begin{eqnarray*}
  \alpha+\alpha^q &=& \operatorname{Tr}^{q^2}_q(\alpha)=1, \\
  b+b^q &=& \operatorname{Tr}^{q^2}_q(b)=\operatorname{Tr}^{q^2}_q(b_1+b_2\alpha)=b_2,\\
  \alpha b^q+\alpha^{q}b &=& \alpha(b+b_2)+(\alpha+1)b=b_1,
\end{eqnarray*}
we  have
\begin{equation}\label{even y}
  y^2+y(c+b_2)+c^2\alpha^{q+1}+cb_1=0.
\end{equation}
\begin{itemize}
  \item If $c=0$, then
  \begin{equation*}
  y^2+yb_2=0.
  \end{equation*}
  So
  \begin{equation*}
    y=0, b_2,
  \end{equation*}
  and then
  \begin{equation*}
    x=1, 1+\frac{b_2}{b}.
  \end{equation*}
  The characteristic $p=2$ implies that $1=-1$, combining  that $ x\in \mu_{q+1}\backslash \{-1\}$, we have that
  the number $N(b,c)$ of solutions to the equation $E(b,c)$ in $\mu_{q+1}\backslash \{-1\}$  is equal to
  \begin{equation*}
 \begin{cases}
0,&{\text{if}}\  b_2=0, \\
{1,}&{\text{if}}\  b_2\ne0.
\end{cases}
\end{equation*}
  \item If $c\ne0,$ then $-1$ is not a solution  to the equation $E(b,c)$.
  \begin{description}
    \item[(I)] If $c+b_2=0,$ i.e. $b_2=c,$ then
    \begin{equation*}
  y^2+c^2\alpha^{q+1}+cb_1=0.
  \end{equation*}
  So
    \begin{equation*}
    y=(c^2\alpha^{q+1}+cb_1)^{1/2},
   \end{equation*}
   and then
   \begin{equation*}
    x=\frac{(c^2\alpha^{q+1}+cb_1)^{1/2}+b+c\alpha}{b}.
   \end{equation*}
   The number $N(b,c)$ of solutions to the equation $E(b,c)$ in $\mu_{q+1}\backslash \{-1\}$  is 1.
    \item[(II)] If $c+b_2\ne0,$ i.e. $b_2\ne c,$ then
    \begin{equation*}
  (\frac{y}{c+b_2})^2+\frac{y}{c+b_2}+\delta=0.
\end{equation*}
By Lemma \ref{Lemma3}, we have that
  the number $N(b,c)$ of solutions to the equation $E(b,c)$ in $\mu_{q+1}\backslash \{-1\}$  is equal to
  \begin{equation*}
 \begin{cases}
0,&{\text{if}}\  \operatorname{Tr}^{q}_2(\delta)=1, \\
{2,}&{\text{if}}\  \operatorname{Tr}^{q}_2(\delta)=0.
\end{cases}
\end{equation*}
  \end{description}
\end{itemize}
\end{proof}

\begin{Proposition}\label{proposition even}
With the same notation as in Theorem \ref{theorem even}, we have that the number of $(b,c)$
 in each case is
\begin{description}
  \item[(i)] 1,   where  $b=0,  c=0;$
  \item[(ii)] $q-1$,  where $b=0, c\ne0;$
  \item[(iii)] $q-1$,  where $b\ne0, c=0,b_2=0;$
  \item[(iv)] $q(q-1)$, where   $b\ne0, c=0,b_2\ne0;$
  \item[(v)] $q(q-1)$,   where   $b\ne0, c\ne0,b_2=c;$
  \item[(vi.1)] $\frac{1}{2}q(q-1)(q-2)$,  where   $b\ne0, c\ne0,b_2\ne c,b_2\ne 0, \operatorname{Tr}^{q}_2(\delta)=1;$
  \item[(vi.2)] $\frac{1}{2}q(q-1)$,  where   $b\ne0,c\ne0,b_2\ne c,b_2=0,\operatorname{Tr}^{q}_2(\delta)=1;$
  \item[(vii.1)] $\frac{1}{2}q(q-1)(q-2)$,  where   $b\ne0, c\ne0,b_2\ne c,b_2\ne 0, \operatorname{Tr}^{q}_2(\delta)=0;$
  \item[(vii.2)] $(\frac{1}{2}q-1)(q-1)$,  where   $b\ne0, c\ne0,b_2\ne c,b_2=0, \operatorname{Tr}^{q}_2(\delta)=0.$
\end{description}
\end{Proposition}
\begin{proof}
The cases (i)-(v) are trivial. Now we assume that $b\ne0, c\ne0,b_2\ne c.$
\begin{itemize}
  \item If $b_2\ne0,$ then $b_1\in \mathbb{F}_{q}$. When $b_2,c\in \mathbb{F}_{q}^*$ are given, the map $\sigma: \mathbb{F}_{q}\rightarrow\mathbb{F}_{q}$, $\sigma(b_1)=\delta$ is a permutation. So there are $\frac{1}{2}q$ $b_1\in \mathbb{F}_{q}$ satisfy that $\operatorname{Tr}^{q}_2(\delta)=1,$ and there are $\frac{1}{2}q$ $b_1\in \mathbb{F}_{q}$ satisfy that $\operatorname{Tr}^{q}_2(\delta)=0.$   Combining that $b_2\ne 0, c\ne0,b_2\ne c,$  the number is
       \begin{equation*}
  \begin{cases}
   \frac{1}{2}q(q-1)(q-2),&{\text{if}}\  \operatorname{Tr}^{q}_2(\delta)=1, \\
   {\frac{1}{2}q(q-1)(q-2),}&{\text{if}}\  \operatorname{Tr}^{q}_2(\delta)=0.
   \end{cases}
  \end{equation*}
  \item If $b_2=0$, then $\delta=\alpha^{q+1}+\frac{b_1}{c}$. Since $\alpha^{q+1}=\alpha^{q}\alpha=(\alpha+1)\alpha=\alpha^2+\alpha$, we have $\operatorname{Tr}^{q}_2(\alpha^{q+1})=\operatorname{Tr}^{q}_2(\alpha^2+\alpha)=0,$ and then $\operatorname{Tr}^{q}_2(\delta)=\operatorname{Tr}^{q}_2(\frac{b_1}{c}).$ Since  $b\ne0$  while $b_2=0$, we must have $b_1\in \mathbb{F}_{q}^*.$  So there are $\frac{1}{2}q$ $b_1\in \mathbb{F}_{q}^*$ satisfy that $\operatorname{Tr}^{q}_2(\frac{b_1}{c})=1,$ and $(\frac{1}{2}q-1)$ $b_1\in \mathbb{F}_{q}^*$ satisfy that $\operatorname{Tr}^{q}_2(\frac{b_1}{c})=0.$   Combining that $c\ne0,$  the number is
      \begin{equation*}
 \begin{cases}
   \frac{1}{2}q(q-1),&{\text{if}}\  \operatorname{Tr}^{q}_2(\delta)=1, \\
   {(\frac{1}{2}q-1)(q-1),}&{\text{if}}\  \operatorname{Tr}^{q}_2(\delta)=0.
\end{cases}
\end{equation*}
\end{itemize}
\end{proof}
%
%

\section{odd characteristic cases} \label{odd characteristic cases}
In this section, we will solve the equations $E(b,c)$ for all $b\in \mathbb{F}_{q^2}, c\in\mathbb{F}_{q}$ in the  odd characteristic cases.
\begin{Theorem}\label{theorem odd}
Let $p$ be an odd prime number and $q$ a power of $p$. Fix $\alpha\in \mathbb{F}_{q^2}\backslash\mathbb{F}_{q}$ such that $\operatorname{Tr}^{q^2}_q(\alpha)=0$. Let $b_1, b_2, c\in\mathbb{F}_{q}$ and $b=b_1+b_2\alpha\in\mathbb{F}_{q^2}$. Define
\begin{equation*}
\Delta=b_2^2-\frac{c^2+4cb_1}{4\alpha^{q+1}}.
\end{equation*}
 Then the number $N(b,c)$ of solutions to the equation $E(b,c)$ in $\mu_{q+1}\backslash \{-1\}$  is equal to
 \begin{description}
   \item[(i)] $q$, if $b=0,  c=0;$
   \item[(ii)] 0, if $b=0, c\ne0;$
   \item[(iii)] 0, if $b\ne0, c=0,b_2=0;$
   \item[(iv)] 1, if $b\ne0, c=0,b_2\ne0;$
   \item[(v)] 1, if $b\ne0, c\ne0, \Delta=0;$
   \item[(vi)]  0, if $b\ne0, c\ne0, \Delta\  {\text{is   not a square in}} \ \mathbb{F}_q^{*};$
   \item[(vii)] 2, if $b\ne0, c\ne0, \Delta\  {\text{is a  square in}} \ \mathbb{F}_q^{*}.$
 \end{description}
\end{Theorem}
\begin{proof}
When $b=0,$ the results are trivial. Next, we let $b\ne0.$  Since $c\in \mathbb{F}_{q}$, we have $c=\operatorname{Tr}^{q^2}_q(\frac{c}{2}).$  We can  rewrite  the equation $E(b,c)$ as
\begin{equation*}
  \operatorname{Tr}^{q^2}_q(bx+b+\frac{c}{2})=0.
\end{equation*}
By Lemma \ref{Lemma2} (ii), there exists $y\in\mathbb{F}_{q}$ such that
\begin{equation*}
bx+b+\frac{c}{2}=y\alpha.
\end{equation*}
Then
\begin{equation}\label{odd x}
  x=\frac{y\alpha-b-\frac{c}{2}}{b},
\end{equation}
and
\begin{equation*}
  x^q=\frac{y\alpha^q-b^q-\frac{c}{2}}{b^q}.
\end{equation*}
Since $x\in \mu_{q+1}\backslash \{-1\}$, we have
\begin{equation*}
  1=x^{q+1}=x\cdot x^q=\frac{y\alpha-b-\frac{c}{2}}{b}\cdot \frac{y\alpha^q-b^q-\frac{c}{2}}{b^q}.
\end{equation*}
we rewrite the above equality as
\begin{equation*}
  y^2\alpha^{q+1}-y(\frac{c}{2}\alpha+\frac{c}{2}\alpha^q+\alpha b^q+\alpha^qb)+\frac{c^2}{4}+\frac{c}{2}(b+b^q)=0.
\end{equation*}
Since
\begin{eqnarray*}
  \alpha+\alpha^q &=& \operatorname{Tr}^{q^2}_q(\alpha)=0, \\
  -\alpha^2 &=& (-\alpha)\alpha=\alpha^{q+1}, \\
  b+b^q &=& \operatorname{Tr}^{q^2}_q(b)=\operatorname{Tr}^{q^2}_q(b_1+b_2\alpha)=2b_1,\\
  \alpha b^q+\alpha^{q}b &=& \alpha(b_1-b_2\alpha)-\alpha(b_1+b_2\alpha)=2b_2\alpha^{q+1},
\end{eqnarray*}
we  have
\begin{equation}\label{odd y}
  y^2-2b_2y+\frac{c^2+4cb_1}{4\alpha^{q+1}}=0.
\end{equation}
\begin{itemize}
  \item If $c=0$, then
  \begin{equation*}
  y^2-2b_2y=0.
  \end{equation*}
  So
  \begin{equation*}
    y=0, 2b_2,
  \end{equation*}
  and then
  \begin{equation*}
    x=-1, -1+\frac{2b_2\alpha}{b}.
  \end{equation*}
  Combining  that $ x\in \mu_{q+1}\backslash \{-1\}$, we have that
  the number $N(b,c)$ of solutions to the equation $E(b,c)$ in $\mu_{q+1}\backslash \{-1\}$  is equal to
  \begin{equation*}
 \begin{cases}
0,&{\text{if}}\  b_2=0, \\
{1,}&{\text{if}}\  b_2\ne0.
\end{cases}
\end{equation*}
  \item If $c\ne0,$ then $-1$ is not a solution  to the equation $E(b,c)$. we rewrite the equation (\ref{odd y}) as
  \begin{equation}
  (y-b_2)^2=\Delta.
\end{equation}
    We have  that the number $N(b,c)$ of solutions to the equation $E(b,c)$ in $\mu_{q+1}\backslash \{-1\}$  is equal to
  \begin{equation*}
 \begin{cases}
1,&{\text{if}}\  \Delta=0, \\
0,&{\text{if}}\  \Delta\  {\text{is   not a square in}} \ \mathbb{F}_q^{*}, \\
{2,}&{\text{if}}\  \Delta\  {\text{is  a square in}} \ \mathbb{F}_q^{*}.
\end{cases}
\end{equation*}
\end{itemize}
\end{proof}

\begin{Proposition} \label{proposition odd}
With the same notation as in Theorem \ref{theorem odd},
 we have that the number of $(b,c)$
   in each case is
\begin{description}
  \item[(i)] 1,    where  $b=0,  c=0;$
  \item[(ii)] $q-1$,  where $b=0, c\ne0;$
  \item[(iii)] $q-1$,  where $b\ne0, c=0,b_2=0;$
  \item[(iv.1)] $(q-1)^2$,  where   $b\ne0, c=0,b_2\ne0,b_1\ne0;$
  \item[(iv.2)] $(q-1)$,  where   $b\ne0, c=0,b_2\ne0,b_1=0;$
  \item[(v.1)]  $q^2-3q+2$, where $b\ne0, c\ne0, b_1\ne0, \Delta=0;$
  \item[(v.2)]  $2q-2$, where $b\ne0, c\ne0, b_1=0, \Delta=0;$
  \item[(vi.1)] $\frac{1}{2}(q-1)(q^2-2q+1)$, where $b\ne0, c\ne0, b_1\ne0,  \Delta\  {\text{is  not a square in}} \ \mathbb{F}_q^{*};$
  \item[(vi.2)] $\frac{1}{2}(q-1)(q-2+\eta(-1))$, where $b\ne0, c\ne0, b_1=0, \Delta\  {\text{is  not a square in}} \ \mathbb{F}_q^{*};$
  \item[(vii.1)] $\frac{1}{2}(q-1)(q^2-2q+3)$, where $b\ne0, c\ne0, b_1\ne0, \Delta\  {\text{is a  square in}} \ \mathbb{F}_q^{*};$
  \item[(vii.2)] $\frac{1}{2}(q-1)(q-4-\eta(-1))$, where $b\ne0, c\ne0, b_1=0, \Delta\  {\text{is a  square in}} \ \mathbb{F}_q^{*}.$
\end{description}
\end{Proposition}
\begin{proof}
The cases (i)-(iv) are trivial. Now we assume that $b\ne0, c\ne0.$
Let
\begin{equation*}
x_1=\frac{b_1}{2\alpha^{(q+1)/2}}, x_2=b_2, x_3=\frac{c}{2\alpha^{(q+1)/2}},
\end{equation*}
and $g(x_1,x_2,x_3)=x_2^2-x_3^2-4x_1x_3.$
Then
\begin{equation*}
g(\frac{b_1}{2\alpha^{(q+1)/2}},b_2,\frac{c}{2\alpha^{(q+1)/2}})=\Delta.
\end{equation*}
\begin{itemize}
  \item If $b_1\ne0,$ then $b_2\in \mathbb{F}_{q}$. In this case, $x_1\in \mathbb{F}_q^{*}, x_2 \in\mathbb{F}_q, x_3\in\mathbb{F}_q^{*}.$  Combining that $(b_1,b_2,c)\in\mathbb{F}_{q}^*\times\mathbb{F}_{q}\times\mathbb{F}_{q}^*$ and Lemma \ref{Lemma6} (i), the number is
   \begin{equation*}
  \begin{cases}
   q^2-3q+2 &, {\text{if}}\   \Delta=0, \\
   \frac{1}{2}(q-1)(q^2-2q+1) &, {\text{if}}\   \Delta\  {\text{is   not a square in}} \ \mathbb{F}_q^{*},\\
   \frac{1}{2}(q-1)(q^2-2q+3) &, {\text{if}}\    \Delta\  {\text{is  a square in}} \ \mathbb{F}_q^{*}.
   \end{cases}
    \end{equation*}
  \item If $b_1=0,$ then $b_2\in \mathbb{F}_{q}^*$. In this case, $x_1=0, x_2 \in\mathbb{F}_q^*, x_3\in\mathbb{F}_q^{*}.$   Combining that $(b_1,b_2,c)\in\{0\}\times\mathbb{F}_{q}^*\times\mathbb{F}_{q}^*$ and  Lemma \ref{Lemma6} (ii), the number is
    \begin{equation*}
  \begin{cases}
   2q-2 &, {\text{if}}\  \Delta=0, \\
   \frac{1}{2}(q-1)(q-2+\eta(-1)) &, {\text{if}}\  \Delta\  {\text{is   not a square in}} \ \mathbb{F}_q^{*},\\
   \frac{1}{2}(q-1)(q-4-\eta(-1)) &, {\text{if}}\  \Delta\  {\text{is  a square in}} \ \mathbb{F}_q^{*}.
   \end{cases}
    \end{equation*}
\end{itemize}
\end{proof}

\section{the weight distributions and an answer to the conjecture} \label{the weight distribution}
 In this section, combining the results about the number of solutions to equations $E(b,c)$  in Section \ref{even characteristic cases}  and   Section \ref{odd characteristic cases},   we give the weight distributions of the linear codes $\widetilde{\overline{C_D}},$ and then  give  an answer to Conjecture 1.

 Let $C$ be an $[n,k,d]_q$ linear code over $\mathbb{F}_q.$  For a codeword $c=(c_1, c_2,\cdots, c_n)\in C$, define its \emph{Hamming weight} as  wt$(c):= |\{1 \leq i\leq n : c_i\neq 0\}|=n-|\{1 \leq i\leq n : c_i=0\}|.$
 Let $A_i$ denote the frequency of the codewords of weight  $i$ in   an $[n,k,d]_q$ linear code $C$, where $0\le i\le n$. Then the sequence $(1,A_1,A_2,\dots,A_n)$ is called the \emph{weight distribution} of $C$.  The
weight distribution not only contains the information of the
capabilities of error detection and correction, but also allows
the computation of the error probability of error detection and
correction of a given code.


  Recall that  $\mu_{q+1}=\{x\in \mathbb{F}_{q^2}: x^{q+1}=1\},$ $D=\mu_{q+1}\backslash \{-1\},$ and  $\widetilde{\overline{C_D}}$ be the linear code defined as in Equation (\ref{code}).

 First, we give the weight distributions of the linear codes $\widetilde{\overline{C_D}}$ in even characteristic cases.
\begin{Theorem}
Let $p=2$ and $q>2$ a power of $p$. Fix $\alpha\in \mathbb{F}_{q^2}\backslash\mathbb{F}_{q}$ such that $\operatorname{Tr}^{q^2}_q(\alpha)=1$. Let $b_1, b_2, c\in\mathbb{F}_{q}$ and $b=b_1+b_2\alpha\in\mathbb{F}_{q^2}$. Then $\widetilde{\overline{C_D}}$ is a $[q+1,3,\ge q-2]_q$ with the weight distribution in Table \ref{table even}.
\end{Theorem}
\begin{proof}
 For $b\in \mathbb{F}_{q^2}$ and $c\in\mathbb{F}_{q},$ recall that the codeword
\begin{equation*}
  c(b,c)=((\operatorname{Tr}^{q^2}_q(bx+b)+c)_{x\in D},-\operatorname{Tr}^{q^2}_q(b)),
\end{equation*}
and  $N(b,c)$ is the number  of solutions to the equation $E(b,c)$ in $\mu_{q+1}\backslash \{-1\}.$ So the Hamming weight $\operatorname{wt}(c(b,c))$ of $c(b,c)$ is equal to
  \begin{equation*}
 \begin{cases}
  q+1-N(b,c),&{\text{if}}\  \operatorname{Tr}^{q^2}_q(b)\ne0, \\
{ q-N(b,c),}&{\text{if}}\  \operatorname{Tr}^{q^2}_q(b)=0.
\end{cases}
\end{equation*}
Since $\operatorname{Tr}^{q^2}_q(b)=\operatorname{Tr}^{q^2}_q(b_1+b_2\alpha)=b_2,$ we have $\operatorname{Tr}^{q^2}_q(b)=0$ if and only if $b_2=0.$ Combining Theorem \ref{theorem even} and Proposition \ref{proposition even} in Section \ref{even characteristic cases}, we have
\begin{enumerate}
  \item $\operatorname{wt}(c(b,c))=q+1$ if and only if $b_2\ne0, N(b,c)=0$. So
  $A_{q+1}=\frac{1}{2}q(q-1)(q-2).$
  \item $\operatorname{wt}(c(b,c))=q$ if and only if $b_2\ne0, N(b,c)=1$ or $b_2=0, N(b,c)=0.$ So $A_q=q(q-1)+q(q-1)+(q-1)+(q-1)+\frac{1}{2}q(q-1)=\frac{1}{2}(q-1)(5q+4).$
  \item $\operatorname{wt}(c(b,c))=q-1$ if and only if $b_2\ne0, N(b,c)=2$ or $b_2=0, N(b,c)=1.$ So $A_{q-1}=\frac{1}{2}q(q-1)(q-2).$
  \item $\operatorname{wt}(c(b,c))=q-2$ if and only if  $b_2=0, N(b,c)=2.$ So $A_{q-2}=\frac{1}{2}(q-2)(q-1).$
  \item $\operatorname{wt}(c(b,c))=0$ if and only if $b_2=0, N(b,c)=q.$ So $A_0=1.$
\end{enumerate}
\end{proof}

\begin{table}
\begin{center}
\caption{the weight distribution of $\widetilde{\overline{C_D}}$ for even $q$.}
\label{table even}
\begin{tabular}{| c | c  |}
   \hline
    Weight & Frequency  \\

   \hline
   $q+1$& $\frac{1}{2}q(q-1)(q-2)$ \\
\hline
   $q$& $\frac{1}{2}(q-1)(5q+4)$ \\
\hline
   $q-1$& $\frac{1}{2}q(q-1)(q-2)$ \\
\hline
     $q-2$& $\frac{1}{2}(q-2)(q-1)$  \\
   \hline
  0& 1 \\
 \hline
\end{tabular}
\end{center}
\end{table}

 Next, we give the weight distributions of the linear codes $\widetilde{\overline{C_D}}$ in odd characteristic cases.

\begin{Theorem}
Let $p$ be an odd prime number and $q$ a power of $p$. Fix $\alpha\in \mathbb{F}_{q^2}\backslash\mathbb{F}_{q}$ such that $\operatorname{Tr}^{q^2}_q(\alpha)=0$. Let $b_1, b_2, c\in\mathbb{F}_{q}$ and $b=b_1+b_2\alpha\in\mathbb{F}_{q^2}$. Then $\widetilde{\overline{C_D}}$ is a $[q+1,3,\ge q-2]_q$ with the  weight distribution in Table \ref{table odd}.
\end{Theorem}
\begin{proof}
 For $b\in \mathbb{F}_{q^2}$ and $c\in\mathbb{F}_{q},$
  the Hamming weight $\operatorname{wt}(c(b,c))$ of $c(b,c)$ is equal to
  \begin{equation*}
 \begin{cases}
  q+1-N(b,c),&{\text{if}}\  \operatorname{Tr}^{q^2}_q(b)\ne0, \\
{ q-N(b,c),}&{\text{if}}\  \operatorname{Tr}^{q^2}_q(b)=0.
\end{cases}
\end{equation*}
Since $\operatorname{Tr}^{q^2}_q(b)=\operatorname{Tr}^{q^2}_q(b_1+b_2\alpha)=2b_1,$ we have $\operatorname{Tr}^{q^2}_q(b)=0$ if and only if $b_1=0.$ Combining Theorem \ref{theorem odd} and Proposition \ref{proposition odd} in Section \ref{odd characteristic cases}, we have
\begin{enumerate}
  \item $\operatorname{wt}(c(b,c))=q+1$ if and only if $b_1\ne0, N(b,c)=0$. So
  $A_{q+1}=\frac{1}{2}(q-1)(q^2-2q+3).$
  \item $\operatorname{wt}(c(b,c))=q$ if and only if $b_1\ne0, N(b,c)=1$ or $b_1=0, N(b,c)=0.$ So $A_q=(q-1)^2+(q^2-3q+2)+(q-1)+\frac{1}{2}(q-1)(q-2+\eta(-1))=\frac{1}{2}(q-1)(5q-6+\eta(-1)).$
  \item $\operatorname{wt}(c(b,c))=q-1$ if and only if $b_1\ne0, N(b,c)=2$ or $b_1=0, N(b,c)=1.$ So $A_{q-1}=\frac{1}{2}(q-1)(q^2-2q+3)+(q-1)+(2q-2)=\frac{1}{2}(q-1)(q^2-2q+9).$
  \item $\operatorname{wt}(c(b,c))=q-2$ if and only if  $b_1=0, N(b,c)=2.$ So $A_{q-2}=\frac{1}{2}(q-1)(q-4-\eta(-1)).$
  \item $\operatorname{wt}(c(b,c))=0$ if and only if $b_1=0, N(b,c)=q.$ So $A_0=1.$
\end{enumerate}
\end{proof}

\begin{table}
\begin{center}
\caption{the weight distribution of $\widetilde{\overline{C_D}}$ for odd $q$.}
\label{table odd}
\begin{tabular}{| c | c  |}
   \hline
    Weight & Frequency  \\

   \hline
   $q+1$& $\frac{1}{2}(q-1)(q^2-2q+3)$ \\
\hline
   $q$& $\frac{1}{2}(q-1)(5q-6+\eta(-1))$ \\
\hline
   $q-1$& $\frac{1}{2}(q-1)(q^2-2q+9)$ \\
\hline
     $q-2$& $\frac{1}{2}(q-1)(q-4-\eta(-1))$  \\
   \hline
  0& 1 \\
 \hline
\end{tabular}
\end{center}
\end{table}

Now we can give an answer to Conjecture \ref{conjecture}.
\begin{Theorem}
Let $p$ be a    prime number and $q>2$ a power of $p$. Let $\widetilde{\overline{C_D}}$ be the linear code defined as in Equation (\ref{code}).
\begin{enumerate}
  \item If $q=3,5$, then $\widetilde{\overline{C_D}}$ is a $[q+1,3,q-1]_q$ MDS code.
  \item If $q\ne3,5$, then $\widetilde{\overline{C_D}}$ is a $[q+1,3,q-2]_q$ NMDS code.
\end{enumerate}
\end{Theorem}
\begin{proof}
If $q$ is even, then $A_{q-2}=\frac{1}{2}(q-2)(q-1)$. So $A_{q-2}\ne0$ if and only if $q>2.$

If $q$ is odd, then $A_{q-2}=\frac{1}{2}(q-1)(q-4-\eta(-1)).$ So $A_{q-2}\ne0$ if and only if $q\ne3,5.$
\begin{enumerate}
  \item If $q=3,5$, then $\widetilde{\overline{C_D}}$ is a $[q+1,3,q-1]_q$   code. Since $n-k+1=q-1=d,$ $\widetilde{\overline{C_D}}$ is an MDS code.
  \item If $q\ne3,5$, then $\widetilde{\overline{C_D}}$ is a $[q+1,3,q-2]_q$  code. Since $n-k+1=q-1=d+1,$ $\widetilde{\overline{C_D}}$ is an AMDS code.
  Assume that the dual $\widetilde{\overline{C_D}}^{\perp}$ is a $[q+1,q-2,d^{\perp}]_q$ code. The Singleton bound of $\widetilde{\overline{C_D}}^{\perp}$ gives that $d^{\perp}\le4$.  Since $\widetilde{\overline{C_D}}$ is not an MDS code, we have that $\widetilde{\overline{C_D}}^{\perp}$ is not an MDS code too. So $d^{\perp}<4$. Let $d=((d_x)_{x\in D}, d_0)\in\widetilde{\overline{C_D}}^{\perp}\backslash \{0\}.$ Then for any  $b\in \mathbb{F}_{q^2}$ and $c\in\mathbb{F}_{q},$ we have
  \begin{equation*}
    \sum_{x\in D}(\operatorname{Tr}^{q^2}_q(bx+b)+c)d_x+(-\operatorname{Tr}^{q^2}_q(b))d_0=0.
  \end{equation*}
  We rewrite the above equality as
  \begin{equation*}
    \operatorname{Tr}^{q^2}_q(b(\sum_{x\in D}xd_x+\sum_{x\in D}d_x-d_0))=-c(\sum_{x\in D}d_x).
  \end{equation*}
  Since  the above equality holds for all $b\in \mathbb{F}_{q^2}$ and $c\in\mathbb{F}_{q},$ we have
  \begin{equation*}
 \begin{cases}
  \sum_{x\in D}xd_x+\sum_{x\in D}d_x-d_0 &= 0 \\
  \ \ \ \ \ \ \ \ \ \ \   \sum_{x\in D}d_x &= 0,
\end{cases}
\end{equation*}
and then
\begin{equation*}
\begin{cases}
\ \ \ \ \    \sum_{x\in D}d_x &= 0\\
  \sum_{x\in D}xd_x-d_0 &= 0.
\end{cases}
\end{equation*}
So $d=((d_x)_{x\in D}, d_0)$ is a solution of  a system of homogeneous linear equations. Any $2\times 2$ submatrix of the coefficient matrix of this system of homogeneous linear equations has the form
  \begin{equation*}
    \left(
       \begin{array}{cc}
         1 & 1 \\
         x_1 & x_2 \\
       \end{array}
     \right), x_1,x_2\in D
  \end{equation*}
  or
  \begin{equation*}
    \left(
       \begin{array}{cc}
         1 & 0 \\
         x & -1 \\
       \end{array}
     \right), x\in D.
  \end{equation*}
  So any two columns of this coefficient matrix are linearly independent, and $\operatorname{wt}(d)\ge 3.$ So $d^{\perp}\ge3$. Combining that $d^{\perp}<4$, we have $d^{\perp}=3$, and $\widetilde{\overline{C_D}}^{\perp}$ is a $[q+1,q-2,3]_q$ AMDS code. Since both $\widetilde{\overline{C_D}}$ and $\widetilde{\overline{C_D}}^{\perp}$ are AMDS codes, we have that $\widetilde{\overline{C_D}}$ is an NMDS code.
\end{enumerate}
\end{proof}

\section{concluding remarks} \label{concluding remarks}
In this paper, we  completely determine the number   of solutions to $E(b,c)$ in $\mu_{q+1}\backslash \{-1\}$ for all $b\in \mathbb{F}_{q^2}, c\in\mathbb{F}_{q}$. As an application, we can give the    weight distributions of the linear codes $\widetilde{\overline{C_D}},$ and    give  an completely  answer to Conjecture 1 given by Heng in \cite{H2023}. We prove that  if $q\ne3,5$, then $\widetilde{\overline{C_D}}$ is a $[q+1,3,q-2]_q$ NMDS code.  We believe that our method  could be helpful to solve other similar interesting problems.

If our aim is   to prove Conjecture 1 only, the process can be much simplified. We need not give the  weight distributions; we only need to prove $A_{q-2}>0.$ For this purpose, Theorem \ref{theorem even} (vii), Proposition \ref{proposition even} (vii.2), Theorem  \ref{theorem odd} (vii), Proposition \ref{proposition odd} (vii.2) are enough.

 {}
\end{document}